\documentclass[11pt]{article} 

\usepackage[utf8]{inputenc} 

\usepackage{geometry} 
\usepackage{graphicx} 

\usepackage{booktabs} 
\usepackage{array} 
\usepackage{paralist} 
\usepackage{verbatim} 
\usepackage{subfig} 
\usepackage{amsmath}
\usepackage{tikz}
\usepackage{pgf}
\usepackage{amsthm}
\usepackage{listings} 
\usepackage{graphicx}

\usepackage{fancyhdr} 
\usepackage{sectsty}
\allsectionsfont{\sffamily\mdseries\upshape} 

\usepackage[nottoc,notlof,notlot]{tocbibind} 
\usepackage[titles,subfigure]{tocloft} 

\title{Accrual valuation and mark to market adjustment}
\author{Alexey Bakshaev\\alex.bakshaev@gmail.com}

\begin{document}
\lstset{language=R, framexleftmargin=5mm, frame=shadowbox, rulesepcolor=\color{gray}}

\maketitle

\begin{abstract} 
This paper provides intuition on the relationship of accrual and mark-to-market valuation for cash and forward interest rate trades. Discounted cash flow valuation is compared to spread-based valuation for forward trades, which explains the trader's view on valuation. This is followed by Taylor series approximation for cash trades, uncovering simple intuition behind accrual valuation and mark-to-market adjustment. It is followed by the PNL example modelled in R. Within the Taylor approximation framework, theta and delta are explained. The concept of deferral is explained taking Forward Rate Agreement (FRA) as an example.
\end{abstract}

\section{Introduction}
\subsection{Types of valuation}
Financial accounting requires that valuation of a trade must at least reflect trade economics. At best, the market value is either known and can be applied directly, or can be derived via certain proxies like comparables or hedges that replicate the trade. The most basic way to value an interest flow receivable in the future is to linearly accrue its value over the waiting time. Consider 1 of notional we lend for a T amount of time at a fixed interest rate of r. We will denote the year fraction since the start of the trade until time t as $\Delta_{0t}$, e.g. $\frac{t}{360}$. Likewise, the year fraction from the start of the trade until the interest payment at time T is $\Delta_{0T}$. We will accrue interest linearlly. Given that the payment amount is $r\Delta_{0T}$ and the linear portion of the period to date t is $\frac{\Delta_{0t}}{\Delta_{0T}}$, the accrued portion of interest may be written as $\frac{\Delta_{0t}}{\Delta_{0T}} *  r\Delta_{0T}$ or simply $r\Delta_{0t}$. The accrual-based PV in that case is:\\
\begin{equation}
PV_{accr} = 1 +  r\Delta_{0t}
\end{equation}

Since market conditions at the moment of pricing will be different to those when the trade was made, the market will price the trade differently to accrual-based valuation. Let's assume that we do discounting with known market rates.  The discounting factor for time t as of $t_0=0$ can be written as (2) where $z_{0,t}$ is the market zero rate.
\begin{equation}
DF_{0t}=\frac{1}{1+z_{0t}\Delta_{0t}}
\end{equation}

\subsection{Cash trades and forward trades}
\subsubsection{Forwards}
Forward trades are trades that start in the future. Continuing with the example above, let's assume that we lend 1 of notional on some future date $t_1$ and get it back with fixed rate rebate r on another future date $t_2$ $(t_0 = 0 < t_1 < t_2)$. We will drop t from notation for the sake of brevity, so the interest paid on date $t_2$ (accrued over the period $t_1$ to $t_2$) is written as  $r\Delta_{1,2}$. In this case, discounted cash flow valuation may be written as:
\begin{equation}
PV_{DCF}=-1*DF_{0,1} + (1 + r\Delta_{1,2})*DF_{0,2}
\end{equation}

Extending (1) for a forward-starting case, we can define accrual-based valuation as:
\begin{equation}
PV_{accr} = 
\begin{cases}
  \phantom{-}0 & \text{if}\ t<t_1 \\
  1 +  r\Delta_{t_1,t}          & \text{if}\ t_1<t<t_2 \\ 
  \phantom{-}0 & \text{if}\ t>=t_2 \\
\end{cases} 
\end{equation}
where $\Delta_{t_1,t} = \frac{t-t_1}{360}$ and $r\Delta_{t_1,t}$ is the accrued amount.\\
This piecewise definition of accrual-based valuation follows from a simple economic principle: if the trade starts in the future ($t_0 < t_1$), we have not yet started to accrue interest, hence, accrual valuation will be giving 0 accrued value by definition. This means that until the accrual period starts, accrual-based valuation will yield zero PNL and flat valuation at $PV_{accr} = 1$. However, this is not the case with the DCF-based apporach. It will be subject to time decay (theta) and  revaluation at market rates (delta) from the trade booking date. It will be explained more clearly in due course of this paper. Another important observation is that accrual-based valuation does not depend on market rates: it is driven only by trade parameters (accrual period length and rate) and the relative position of valuation date with respect to accrual period.\\
Interest is assumed to be paid on trade maturity date, so accrual goes to 0 at this moment. Since the principal is also paid at maturity, entire accrual valuation (1) will go to 0 at and past the maturity date. This is equivalent to transforming a security asset to cash, so we don't have the security on the balance anymore, just cash. DCF valuation also goes to zero since we assume that the discounting factor on the maturity date is 0 ($DF_{t=T}=0$). Below is the depiction of the accrued amount in time.\\

\begin{tikzpicture}
\draw[->] (0,0) -- (6,0) node[anchor=north] {$Time$};
\draw	(0,0) node[anchor=north] {0}
		(2,0) node[anchor=north] {1}
		(4,0) node[anchor=north] {2};
\draw	(1,3.5) node{{\scriptsize Forward trade}}
		(4,3.5) node{{\scriptsize Cash trade}};

\draw[->] (0,0) -- (0,4) node[anchor=east] {$Accrual$};
\draw (-1,2) node {$r\Delta_{12}$};
\draw[dotted] (2,0) -- (2,4);

\draw[thick] (0,0) -- (2,0) -- (4, 2) -- (4,0);


\end{tikzpicture}

\subsubsection{"Cash" trades}
"Cash" trades are the trades that are already started ($t_1<=t$). Accrual-based valuation is in line with (1) and (4). We may re-write DCF valuation as:
\begin{equation}
PV_{DCF}= (1 + r\Delta_{1,2})*DF_{t,2}
\end{equation}
\newtheorem{thm}{Theorem}
\newtheorem{lem}[thm]{Lemma}
\newtheorem{example}{Numerical example}

\section{Forwards valuation}
Following the fixed rate example above, let's prove a simple lemma that introduces the notion of equivalence of DCF and spread-based valuations. LHS is a DCF valuation formula (3), RHS is a spread-based valuation formula $PV_{spread-based}=[r - z_{12}]*\Delta_{1,2}*DF_{0,2}$
\begin{lem}
DCF valuation and spread-based valuations are equivalent:
\begin{equation}
-1*DF_{0,1} + (1 + r\Delta_{1,2})*DF_{0,2} \equiv [r - z_{12}]*\Delta_{1,2}*DF_{0,2}
\end{equation}
\end{lem}
\begin{proof}
Let $z_{12}$ be the forward rate from the start date to maturity. By definition:\\
\begin{equation}
DF_{1,2}=\frac{1}{1+z_{12}\Delta_{12}}
\end{equation}
At the same time, $DF_{1,2} = \frac{DF_{0,2}}{DF_{0,1}}$. From this follows that $1+z_{12}\Delta_{12} =  \frac{DF_{0,1}}{DF_{0,2}}$. Substituting $DF_{0,1}$ in (3) we get:\\
\begin{equation}
PV_{DCF}= -1 * (1+z_{12}\Delta_{12}) * DF_{0,2} + [1+r_{12}\Delta_{12}] * DF_{0,2} = [r - z_{12}]*\Delta_{1,2}*DF_{0,2}
\end{equation}
\end{proof}
\subsection{Taylor approximation}
Taylor approximation for $x_0 = 0$ can be written as:
\begin{equation}
\frac{1}{1+x}\sim1-x+o(x)
\end{equation}\\
Hence, the discounting factor may be written as:
\begin{equation}
DF_{0t}=\frac{1}{1+z_{0t}\Delta_{0t}} \sim 1 - z_{0t}\Delta_{0t} + o(z_{0t}\Delta_{0t})
\end{equation}\\
Applying Taylor approximation for the RHS of (6) we get:
\begin{equation}
PV_{spread-based}=[r - z_{12}]*\Delta_{1,2}*DF_{0,2} \sim [r - z_{12}]*\Delta_{1,2} + o(z\Delta)
\end{equation}\\
\begin{example}
Trader's view of forward pricing
\end{example}
Consider a tomorrow-next (TN) interest at maturity trade with the nominal of 1000000 done at 500bps. Let the TN market rate be 300bps, ON market rate be 200 bps. What is the value?\\
We would use Taylor-approximation of the spead based valuation that we explained in (11):\\
 \begin{equation}
PV_{spread-based}\sim [r - z_{12}]*\Delta_{1,2}
\end{equation}\\
Spread-based (following (11)): $1000000 * (0.05 - 0.03) * 1 / 360 = 55.55$.\\
Let's now see how far it is from the actual DCF valuation. $DF_{0,1}=1/(1+0.02*1/360)=0.999944448$, $DF_{1,2}=1/(1+0.03*1/360)=0.999916674$, so $DF_{0,2}=DF_{0,1}*DF_{1,2}=0.999861126$. Using (3) we get: 
\begin{equation}
PV_{DCF}=-1000*0.999944448 + 1000 * (1 + 0.05*1/360) * 0.999861126 = 55.5478
\end{equation}
We see that it is in good coherence with Taylor approximation.

\section{Relationship of accrual based and DCF valuation}
Let's now see what happens with DCF valuation when the trade stops being forward and becomes "cash" so we begin to accrue interest. We will continue with the same fixed rate one-period trade example.
\begin{lem}
Within the first-order Taylor series approximation, DCF valuation may be presented as the sum of the accrual-based PV and the mark-to-market adjustment.
\begin{equation}
PV_{DCF}= (1 + r\Delta_{0,T})*DF_{t,T}  \sim 1 + r\Delta_{0,t} + (r-z_{t,T})\Delta_{t,T} + o(z)
\end{equation}
\end{lem}
\begin{proof}
Let's denote the valuation date as $t$, the maturity date as $T$, the trade rate as $r$ and the zero rate from t to T $z_{t,T}$ as $z$ .
\begin{equation}
\begin{split}
& \frac{1+r\Delta_{0,T}}{1+z\Delta_{t,T}}=\frac{1+r(\Delta_{0,t}+\Delta_{t,T})}{1+z\Delta_{t,T}}=\frac{r\Delta_{0,t}(1+z\Delta_{t,T})-r\Delta_{0,t}z\Delta_{t,T}+1+r\Delta_{t,T}}{1+z\Delta_{t,T}}\\
&=\underbrace{r\Delta_{0,t}}_{\text{Accrual}} + \underbrace{\frac{1+r\Delta_{t,T}-r\Delta_{0,t}z\Delta_{t,T}}{1+z\Delta_{t,T}}}_{\text{MtM Adjustment}}
\end{split}
\end{equation}
Let's now apply Taylor approximation (9) to the MtM component, removing the second-order components. Note that $rz\Delta_{0,t}\Delta_{t,T}$ is also a second-order component.
\begin{equation}
\begin{split}
& r\Delta_{0,t} + \frac{1+r\Delta_{t,T}-r\Delta_{0,t}z\Delta_{t,T}}{1+z\Delta_{t,T}} \sim r\Delta_{0,t} + (1 + r\Delta_{t,T} - rz\Delta_{0,t}\Delta_{t,T})(1-z\Delta_{t,T}) \sim \\
& r\Delta_{0,t} + 1 +  r\Delta_{t,T} - rz\Delta_{0,t}\Delta_{t,T} - z\Delta_{t,T} - rz(\Delta_{t,T})^2 + rz^2\Delta_{0,t}(\Delta_{t,T})^2 \sim \\ 
& \underbrace{1 + r\Delta_{0,t}}_{\text{$PV_{accr} (1)$}} + \underbrace{(r - z)\Delta_{t,T}}_{\text{MtM Adjustment (11)}} + o(z\Delta)
\end{split}
\end{equation}
\begin{equation}
MtMAdjustment \sim (r - z)\Delta_{t,T}
\end{equation}
\end{proof}
As shown above, DCF valuation is equivalent to the sum of accrual-based valuation (1) and mark-to-market adjustment (17). Note that the \textbf{\textit{spread-based valuation}} (11) applied from the valuation date $t$ until the expiration date $T$ here represents \textbf{\textit{mark-to-market adjustment}} (17) that brings the accrual-based valuation in line with the market rate $z$.\\
If the market rate $z$ is equal to the trade rate $r$, the MtM adjustment component cancels out. This draws some parallels to bond markets vocabulary. From the clean price perspective, the trade is priced at par in this case. From the dirty price perspective, the trade is priced at par plus accrued. Once the trade rate is different to the market rate, the clean price becomes affected by the mark-to-market adjustment.\\
\section{Greeks}
Based on (16) let's derive the basic Greeks that define the PNL. Recollect that $\Delta_{0,t}=\frac{t-t_0}{360}$, $\Delta_{t,T}=\frac{T-t}{360}$\\
\begin{equation}
\begin{split}
& \vartheta = \frac{\partial PV}{\partial t} \sim \frac{r}{360} - \frac{r}{360} + \frac{z}{360} \sim \frac{z}{360}\\
& \varrho = \frac{\partial PV}{\partial z} \sim -\Delta_{t,T}
\end{split}
\end{equation}
\begin{itemize}
  \item $\vartheta$ ("theta") addresses sensitivity to time decay $t$. First-order dependency is on the market rate. It shows the market cost of carrying 1 day of notional over to the next day. We may as well re-write $\vartheta$ as the sum of accrued and mtm components:\\
  \begin{equation}
  \vartheta \sim \frac{z}{360} \sim \underbrace{\frac{r}{360}}_{\text{Accrued}} + \underbrace{(\frac{z}{360} - \frac{r}{360})}_{\text{MtM adjustment}}
  \end{equation} 
  \item $\varrho$ ("rho") addresses sensitivity to market rates $z$. We have an inverse relationship on the remaining duration of the trade. This means that the less time is left until expiry, the less sensitive is the price to changes in the market rate. The minus sign tells us that PV decreases as the market rate increases.
\end{itemize}
\subsection{Accrual valuation and MtM adjustment in R}
R is a scientific language popular in many areas. One of its advantages is the easiness of working with vectors. We will populate an array of increasing daycount in the $days$ object, decreasing daycount in the $daysRemaining$ object and will then calculate:\\
 \begin{tabbing}
\textbf{\textit{discFact}} \= - \= a vector of daily discounting factors for each day in $daysRemaining$ object\\
\textbf{\textit{accrued}} \= - \= a vector of daily accrued amounts\\
\textbf{\textit{PV}} \= - \= a vector of DCF-based PVs (one PV per day)\\
\textbf{\textit{mtmAdj}} \= - \=  a vector of mark-to-market adjustments $(11)$\\
\textbf{\textit{PVTaylor}} \= - \=  a vector of Taylor-approximated PVs, see $(16)$\\
\textbf{\textit{unexplained}} \= - \= a vector of the difference between DCF PV and Taylor-approx PV\\
 \end{tabbing}
We will assume the market rate of 700 bps which is intentionally higher than the trade rate of 500 bps (1M of notional for 10 days) to get the negative mtm adjustment (it is more instructive this way). This market rate of 700 bps will be constant each day for the sake of simplicity. Below is the R code that calculates valuation for the above.\\
\begin{lstlisting}
days <- seq(1,10)
daysRemaining <- seq (9, 0, -1)
accrualFractions <- days / 360
fractionsRemaining <- daysRemaining / 360
discFact <- 1/(1+ 0.07 * daysRemaining / 360)
accrued <- 1000000 * 0.05 * accrualFractions
PV <- 1000000 * (1 + 0.05 * 10/360) * discFact
mtmAdj <- 1000000 * (0.05 - 0.07) * fractionsRemaining
PV_Taylor = 1000000 + accrued + mtmAdj
unexplained <- PV - PV_Taylor
PVs <- cbind(days, PV, accrued, mtmAdj, PV_Taylor, unexplained)
PVs
\end{lstlisting}
Executing this code in the R environment yields the following PNL simulation:\\
\includegraphics{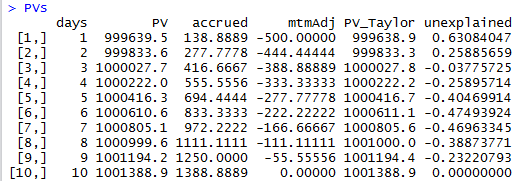}
\\
This example provides some nice intuition:
\begin{itemize}
  \item Accrual linearly increasing which reflects the coupon amount being linearly recognized
  \item Negative MtM adjustment reflects the fact that the market rate of 700bps is higher than the trade rate of 500 bps. It decreases with time as the valuation date $t$ gets closer to the expiration date $T$, since the mtm component is scaled by the duration multiple of $\Delta_{t,T}$. This is in line with sensitivity $\varrho$ (duration) to the market rate $z$ which is exactly that: $ -\Delta_{t,T}$
  \item Since the market rate is not changing, there is no $\varrho$ attribution. The entire daily PnL is explained by time decay $\vartheta \sim 1000000 * 0.07 / 360 \sim 194.44$ (18). This highlights a very important point. While PNL attribution is purely theta-driven, it might be split into the accrual component and the mark-to-market component which depend on the trade rate and the curve rate, respectively. Following (19) we get the following theta components:\\
  \begin{equation}
  \begin{split}
  & {\vartheta}_{accr} \sim \frac{r}{360} \sim 1000000 * 0.05 / 360 \sim 138.89 \\
  & {\vartheta}_{MtM} \sim \frac{z - r}{360} \sim 1000000 * (0.07 - 0.05) / 360 \sim 55.56
  \end{split}
  \end{equation}
      
\end{itemize}

\section {Generalizaions}
\subsection {Floating rates}
Let's now assume the valuation of a single-period forward-starting floating trade. In the formula (8) we need to replace the fixed trade rate $r$ with a floating rate $f$ and a fixed spread $s$. Taylor approximation will be written as:\\
\begin{equation}
PV_{DCF}= [f_{12} + s - z_{12}]*\Delta_{1,2}*DF_{0,2} \sim [f_{12} + s - z_{12}]*\Delta_{1,2}
\end{equation}
If forward rates $f_{12}$ and $z_{12}$ are forecasted from the same yield curve (e.g. LIBOR) they will cancel out and valuation is approximated as simply as 
\begin{equation}
PV_{DCF}=s\Delta_{1,2}
\end{equation}
In the case of zero spread $s=0$, so the valuation of such forward period is 0. This is easily understood, as we both forecast and discount over the period with the same rate. The amount we pay at time 1 (notional amount) will be equal to the discounted (to time 1) value of the flow we receive at point 2, so $-1 + \frac{1 + z_{12}\Delta_{1,2}}{1+z_{12}\Delta_{1,2}}$=0.\\
In the case of LIBOR forecasting and OIS discounting we will get PV depending on LIBOR to OIS spread $f_{12}-z_{12}$\\
Once the rate is fixed $f_{12}=r$, valuation will be equivalent to (14).

\subsection {Applicable products}
Intuition on accrual valuation and mark-to-market adjustment provided in this paper applies to a wide range of interest rate products such as FRA, repo (considering just the cash side of the transaction), interest at maturity, as well as bonds and interest rate swaps. An extreme case are the products where we assume zero duration, e.g. call accounts. For those, the most conservative valuation would be purely accrual-based (if we assume that the money can be withdrawn on the valuation date) with no mark-to-market adjustment.\\
Another extreme case is forward rate agreements (FRAs). Those are forward trades that pay at the beginning of the period. While they stay forward, valuation is driven by the mark-to-market adjustment. However, once a FRA pays, the MtM adjustment component disappears entirelly. Since we received the interest at the beginning of the period, we need to defer it. The explaination of that is provided in the next paragraph.

\subsection {Forward Rate Agreements}
Forward rate agreement (FRA) is a financial contract where the lender (seller) agrees to receive the fixed "contractual" interest rate $r$ in the future (a fixed rate is locked-in now) in exchange for a floating payment in the future (becomes known at the future point $t_1$). A special property of FRA is that it pays at the beginning of the forward period. Let's again denote the beginning of the forward period $t_1$ as 1 and the end $t_2$ as 2. For the sake of simplicity, let's assume that both the fixing and the payment happen at $t_1$. Since the payment is at $t_1$, there is a need to lock-in the discounting rate from 2 to 1. The convention is to assume that it is the same forward rate $f_{12}$ that the period is forecasted (and later fixed) with. Thus, PV of a FRA at point $t_1$ is:
\begin{equation}
PV^{t1}_{FRA} = \frac{[r - f_{12}]*\Delta_{1,2}}{1 +f_{12}\Delta_{1,2}}
\end{equation}
Discounting to time 0 we get:
\begin{equation}
PV_{FRA} = \frac{[r - f_{12}]*\Delta_{1,2}}{1 +f_{12}\Delta_{1,2}}*DF_{0,1}
\end{equation}
Applying Taylor approximation (9) to this formula gives us the same approximation as in (12), which tells us that FRA valuation is a purely mark-to-market adjustment - driven (17):
\begin{equation}
PV_{FRA} \sim [r - f_{12}]*\Delta_{1,2}
\end{equation}

\subsection {Notion of NPV}
To better understand the concept of deferral, we need to introduce the notion of NPV first. NPV is the sum of the outstanding cash position and the present value of the trade. The graph below shows what accrual-based PV, Cash position and NPV look like for the one-period example that we used in (4).\\
\begin{tikzpicture}
\draw[->] (0,0) -- (6,0) node[anchor=north] {$Time$};
\draw	(0,0) node[anchor=north] {0}
(2,0) node[anchor=north] {1}
(4,0) node[anchor=north] {2};

\draw[->] (0,0) -- (0,3.3) node[anchor=east] {$PV_{accr}$};
\draw (-0.2,1) node {1};
\draw (-1.5,1.5) node {1 + $r\Delta_{12}$};

\draw[thick] (0,0)-- (2,0) -- (2,1) -- (4, 1.5) -- (4,0);
\end{tikzpicture}

\begin{tikzpicture}
\draw[->] (0,0) -- (6,0) node[anchor=north] {$Time$};
\draw	(-0.2,0) node[anchor=north] {0}
(1.8,0) node[anchor=north] {1}
(3.8,0) node[anchor=north] {2};

\draw[->] (0,-2) -- (0,1.5) node[anchor=east] {$CashPos$};
\draw (-0.2,-1) node {-1};
\draw (-1.3,0.5) node {$r\Delta_{12}$};

\draw[thick] (0,0)-- (2,0) -- (2,-1) -- (4, -1) -- (4,0.5) -- (5,0.5);
\end{tikzpicture}

\begin{tikzpicture}
\draw[->] (0,0) -- (6,0) node[anchor=north] {$Time$};
\draw	(-0.2,0) node[anchor=north] {0}
(2,0) node[anchor=north] {1}
(4,0) node[anchor=north] {2};

\draw[->] (0,0) -- (0,2) node[anchor=east] {$NPV$};
\draw (-1.3,0.5) node {$r\Delta_{12}$};

\draw[thick] (0,0)-- (2,0) -- (4,0.5) -- (5,0.5);
\end{tikzpicture}\\
As can be seen, NPV gives the notion of linarly increasing interest income that remains constant after the coupon payment date is achived.

\subsection {Deferral}
With FRA we would like to linearly recognize the income over the FRA period. This means that we need to offset the interest payment received at the beginning of the period with the deferral amount so that the net of those would equate to linearly increasing accrual. We can define deferral as:
\begin{equation}
Deferral = 
\begin{cases}
\phantom{-}0 & \text{if}\ t<t_1 \\
-r\Delta_{t,T}          & \text{if}\ t_1<t<t_2 \\ 
\phantom{-}0 & \text{if}\ t>=t_2 \\
\end{cases} 
\end{equation}
where $-r\Delta_{t,T}$ is $-r\frac{T-t}{360}$.\\
It can be seen that the sum of the coupon and deferral would form the profile equivalent to linear accrual. 
\begin{equation}
Deferral + Coupon = Accrual
\end{equation}
The graph below depicts that.\\
\begin{tikzpicture}
\draw[->] (0,0) -- (4,0) node[anchor=north] {$Time$};
\draw	(-0.2,0) node[anchor=north] {1}
(2,0) node[anchor=north] {2};

\draw[->] (0,-1) -- (0,1) node[anchor=east] {$Deferral$};
\draw (-1,-0.5) node {-$r\Delta_{12}$};

\draw[thick] (0,-0.5)-- (2,0) -- (4,0);
\end{tikzpicture}
\begin{tikzpicture}
\draw[->] (0,0) -- (4,0) node[anchor=north] {$Time$};
\draw	(-0.2,0) node[anchor=north] {1}
(2,0) node[anchor=north] {2};

\draw[->] (0,-1) -- (0,1) node[anchor=east] {$CashPos$};
\draw (-1,0.5) node {$r\Delta_{12}$};

\draw[thick] (0,0.5)-- (2,0.5) -- (4,0.5);
\end{tikzpicture}
\\
\begin{tikzpicture}
\draw[->] (0,0) -- (4,0) node[anchor=north] {$Time$};
\draw	(-0.2,0) node[anchor=north] {1}
(2,0) node[anchor=north] {2};

\draw[->] (0,-1) -- (0,1) node[anchor=east] {$Income$};
\draw (-1,0.5) node {$r\Delta_{12}$};

\draw[thick] (0,0)-- (2,0.5) -- (4,0.5);
\end{tikzpicture}\\
\section {Conclusion}
Using Taylor approximation helps to bridge the gap between accrual and mark-to-market valuation. We can see that PV may be split into accrual and mark-to-market adjustment components. Similarly, theta can be viewed as the sum of the accrued and MtM component.\\
We see that the forward valuation formula (6) turns into the mark-to-market adjustment for trades that began to accrue interest. The less time remains until trade expiration, the less effect mark-to-market effect has on valuation. This makes sense since the risk (duration) decreases with time to maturity.

\section {Literature}
\begin{itemize}
\item {Options, futures, and other derivatives / John C. Hull}
\item {Fixed Income Securities: Tools for Today's Markets 3rd Edition / Bruce Tuckman, Angel Serrat}
\end{itemize}
\end{document}